\documentclass{acm_proc_article-sp}

\usepackage{float}
\usepackage{graphics}
\usepackage{graphicx}
\usepackage[thinlines,thiklines]{easybmat}
\usepackage{multirow}
\usepackage{algorithmic}
\usepackage{algorithm}

\newcommand{\N}{\mathbb{N}} 
\newcommand{\Z}{\mathbb{Z}} 
\newcommand{\Q}{\mathbb{Q}} 
\newcommand{\R}{\mathbb{R}} 
\newcommand{\C}{\mathbb{C}} 
\newcommand{\K}{K} 
\newcommand{\p}{\mathfrak{p}} 

\newcommand{\OK}{\mathcal{O}_{\K}}
\newcommand{\Nm}{\mathcal{N}}
\newcommand{\ag}{\mathfrak{a}}
\newcommand{\bg}{\mathfrak{b}}
\newcommand{\cg}{\mathfrak{c}}
\newcommand{\dg}{\mathfrak{d}}
\newcommand{\g}{\mathfrak{g}}

\newcommand{\Mlt}{\mathcal{M}}

\newtheorem{lemma}{Lemma}

\newtheorem{proposition}{Proposition}
\newtheorem{corollary}{Corollary}

\renewcommand{\algorithmicrequire}{\textbf{Input:}}
\renewcommand{\algorithmicensure}{\textbf{Output:}}

\begin{document}

\title{A polynomial time algorithm for computing the HNF of a module over the integers of a number field}

\numberofauthors{2} 
%
\author{
\alignauthor
Jean-Fran\c{c}ois Biasse\\
       \affaddr{Departement of computer science}\\
       \affaddr{2500 University Drive NW}\\
       \affaddr{Calgary Alberta T2N 1N4}\\
       \email{biasse@lix.polytechnique.fr}
\alignauthor
Claus Fieker\\
       \affaddr{Fachbereich Mathematik}\\
       \affaddr{Universit\"{a}t Kaiserslautern}\\
       \affaddr{Postfach 3049}\\
       \affaddr{67653 Kaiserslautern - Germany}\\
       \email{fieker@mathematik.uni-kl.de}
}


\maketitle
\begin{abstract}
We present a variation of the modular algorithm for computing the Hermite Normal Form of an $\OK$-module presented by 
Cohen~\cite{cohen2}, where $\OK$ is the ring of integers of a number field K. The modular strategy was 
conjectured to run in polynomial time by Cohen, but so far, 
no such proof was available in the literature. In this paper, we provide a new method to prevent the 
coefficient explosion and we rigorously assess its complexity with respect to the size of the input and the invariants 
of the field K.
\end{abstract}




\section{Introduction}
The construction of a good basis of an $\OK$-module, where $K$ is a number field and $\OK$ its ring of 
integers, has recently received a growing interest from the cryptographic community. Indeed, $\OK$-modules 
occur in lattice-based cryptography~\cite{LyuMic06icalp,Mic02cyclic,Mic07cyclic, RosenTCC,SSTX09}, where 
cryptosystems rely on the difficulty to find a short element of a module, or solving the closest 
vector problem. The computation of a good basis is crucial for solving these problems, and most of the 
algorithms for computing a reduced basis of a $\Z$-lattice have an equivalent for $\OK$-modules. However, 
applying the available tools over $\Z$ to $\OK$-modules would result in the loss of of their structure. 

The computation of a Hermite Normal Form (HNF)-basis was generalized to $\OK$-modules by Cohen~\cite[Chap. 1]{cohen2}. 
His algorithm returns a basis that enjoys similar properties as the HNF of a $\Z$-module. A 
modular version of this algorithm is conjectured to run in polynomial time, although this statement is not 
proven (see last remark of~\cite[1.6.1]{cohen2}). In addition, Fieker and Stehl{\'e}'s recent algorithm for computing a 
sized-reduced basis relies on the conjectured possibility to compute an HNF-basis for an $\OK$-module in polynomial 
time~\cite[Th. 1]{stehle_fieker_LLL}. This allows a polynomial time equivalent of the LLL algorithm preserving the 
structure of $\OK$-module. In this paper, we adress the problem of the polynomiality of the computation of an 
HNF basis for an $\OK$-module by presenting a modified version of Cohen's algorithm~\cite[Chap. 1]{cohen2}. We 
thus assure the validity of the LLL algorithm for $\OK$-modules of Fieker and Stehl{\'e}~\cite{stehle_fieker_LLL} which 
has applications in lattice-based cryptography, as well as in representations of matrix groups~\cite{fieker_rep} and 
in automorphism algebras of Abelian varieties. In addition, our HNF algorithm allows to compute a basis for 
the intersection of $\OK$ modules, which has applications in list decoding codes based on number fields 
(see~\cite{guruswami_nb_fld} for their description).
\paragraph*{Our contribution} 

We present in this paper the first polynomial time algorithm for computing an HNF basis of an $\OK$-module based on the modular 
approach of Cohen~\cite[Chap. 1]{cohen2}. We rigorously adress its correctness and derive bounds on its run time with 
respect to the size of the input, the dimension of the module and the invariants of the field.

\section{Generalities on number fields}

Let $K$ be a number field of degree $d$. It has $r_1\leq d$ real embeddings $(\sigma_i)_{i\leq r_1}$ and $2r_2$ complex 
embeddings $(\sigma_i)_{r_1 < i \leq 2r_2}$ (coming as $r_2$ pairs of conjugates). The field $K$ is isomorphic to 
$\OK\otimes\Q$ where $\OK$ denotes the ring of integers of $K$. We can embed $K$ in 
$$K_\R := K\otimes \R \simeq \R^{r_1}\times \C^{r_2}, $$ 
and extend the $\sigma_i$'s to $K_\R$. Let $T_2$ be the Hermitian form on $K_\R$ defined by 
$$T_2(x,x') := \sum_i \sigma_i(x)\overline{\sigma_i}(x'),$$
and let $\| x\| := \sqrt{T_2(x,x)}$ be the corresponding $L_2$-norm. Let $(\alpha_i)_{i\leq d}$ such that 
$\OK = \oplus_i \Z\alpha_i$, then the discriminant of $K$ is given by $\Delta_K = \det^2(T_2(\alpha_i,\alpha_j))$. 
The norm of an element $x\in K$ is defined by $\Nm(x) = \prod_i|\sigma_i(x)|$.

To represent $\OK$-modules, we rely on a generalization of the 
notion of ideal, namely the fractional ideals of $\OK$. They can be defined as finitely generated 
$\Z$-modules of $K$. When a fractional ideal is contained in $\OK$, we refer to it as an integral ideal, which is in 
fact an ideal of $\OK$. Otherwise, for every fractional ideal $I$ of $\OK$, there exists $r\in\Z_{>0}$ such that $rI$ is integral. 
The sum and product of two fractional ideals of $\OK$ is given by 
\begin{align*}
IJ &= \{ i_1j_1 + \cdots + i_lj_l\mid l\in \N, i_1,\cdots i_l\in I, j_1,\cdots j_l\in J\}\\
I + J &= \{ i + j\mid i\in I , j\in J\}.
\end{align*}
The fractional ideals of $\OK$ are invertible, that is for every fractional ideal $I$, there exists 
$I^{-1}:= \{ x\in K\mid xI\subseteq \OK\}$ such that $II^{-1} = \OK$. The set of fractional ideals is equipped with a 
norm function defined by $\Nm(I) = \det(M^I)/\det(\OK)$ where the rows of $M^I$ are a $\Z$-basis of $I$. The norm of ideals is multiplicative, and in the case of an 
integral ideal, we have $\Nm(I) = |\OK / I|$. Also note that the norm of $x\in K$ is precisely the norm of the principal 
ideal $(x) = x\OK$. Algorithms for ideal arithmetic in polynomial time are described in Section~\ref{sec:cost}.

\section{The HNF}

Let $M\subseteq K^l$ be a finitely generated $\OK$-module. As in~\cite[Chap. 1]{cohen2}, we say that 
$[(a_i),(\mathfrak{a}_i)]_{i\leq n}$, where $a_i\in K$ and $\mathfrak{a}_i$ is a fractional ideal, is a 
pseudo-basis for $M$ if
$$M = \mathfrak{a}_1a_1\oplus \cdots \oplus \mathfrak{a}_na_n.$$
Note that a pseudo-basis is not unique, and the main result of~\cite{stehle_fieker_LLL} is precisely to compute a 
pseudo-basis of short elements. If the sum is not direct, we call $[(a_i),(\mathfrak{a}_i)]_{i\leq n}$ a pseudo-generating 
set for $M$. Once a pseudo-generating set $[(a_i),(\mathfrak{a}_i)]_{i\leq n}$ for $M$ is known, we can associate a 
pseudo-matrix $A = (A,I)$ to $M$, where $A\in K^{n\times l}$ and $I = (\ag_i)_{i\leq n}$ is a list of $n$ fractional 
ideals such that 
$$M = \ag_1 A_1 + \cdots +\ag_n A_n,$$
where $A_i\in K^l$ is the $i$-th row of $A$. We can construct a pseudo-basis from a pseudo-generating set by using 
the Hermite normal form (HNF) over Dedekind domains (see ~\cite[Th. 1.4.6]{cohen2}). Note that this canonical form is 
also refered to as the pseudo-HNF in~\cite[1.4]{cohen2}. In this paper we simply call it HNF, but we implicitly refer 
to the standard HNF over $\Z$ when dealing with an integer matrix. Assume $A$ is of rank $l$ 
(in particular $n\geq l$), then there exists an $n\times n$ matrix $U = (u_{i,j})$ and $n$ non-zero 
ideals $\bg_1, \cdots , \bg_n$ satisfying 
\begin{enumerate}
 \item $\forall i,j, u_{i,j}\in \bg_i^{-1}\ag_j$.
 \item $\ag = \det(U)\bg$ for $\ag = \prod_i\ag_i$ and $\bg = \prod_i \bg_i$.
 \item The matrix $UA$ is of the form\\
\[ UA = \left( 
   \begin{BMAT}(@)[0.5pt,2cm,2cm]{c}{c.c}
   \begin{BMAT}(e)[1pt,1cm,1cm]{cccc}{cccc}
1      & 0      & \hdots & 0      \\
\vdots & 1      & \ddots & \vdots \\
\vdots & \vdots & \ddots & 0      \\
*      & *      & \hdots & 1\\
  \end{BMAT} \\
\begin{BMAT}[0.5pt,1cm,1cm]{c}{c} 
        (0)
\end{BMAT}
\end{BMAT}
   \right).
\]

 \item $M = \bg_{1}\omega_1\oplus \cdots \oplus \bg_{l}\omega_l$ where $\omega_1,\cdots \omega_l$ are the first $l$ rows of $UA$.
\end{enumerate}

In general, the algorithm of~\cite{cohen2} for computing the HNF of a pseudo-matrix takes exponential time, 
but as in the integer case, there exists a modular one which is polynomial in the dimensions of $A$, the 
degree of $K$, and the bit size of the modulus. Note that in the case of a pseudo matrix representing an 
$\OK$-module $M$, the modulus is an integral multiple of the determinantal ideal $\g(M)$, which is generated 
by all the ideals of the form 
$$\det_{i_1,\cdots,i_l}(A)\cdot \ag_{i_1}\cdots\ag_{i_l},$$
where $\det_{i_1,\cdots,i_l}(A)$ is the determinant of the $l\times l$ minor consisting of the last $l$ columns 
of rows of indices $i_1,\cdots,i_l$. The determinantal ideal is a rather involved structure, except in the case $l = n$. 
In applications, the modulus is frequently known. In the rest of the paper, we restrict ourselves to the case of 
an $n\times n$ matrix $A$ of rank $n$. One can immediatly derive polynonmial time algorithms for the rectangular 
case, and for the case of a singular matrix $A$.

\section{Notion of size}

To ensure that our algorithm for computing an HNF basis of an $\OK$-module runs in polynomial time, we need a notion 
of size 
that bounds the bit size required to represent ideals and field elements. An ideal 
$I\subseteq \OK$ is given by the matrix $M^I\in\Z^{d\times d}$ of its basis expressed in an integral basis 
$\omega_1,\cdots,\omega_d$ of $\OK$. If the matrix is in Hermite Normal Form, the size required to store it is therefore bounded by 
$d^2\max_{i,j}\left(\log(|M^I_{i,j}|)\right)$, where $\log(x)$ is the base 2 logarithm of $x$. In the meantime, every coefficient of 
$M^I$ is bounded by $|\det(M^I)|=\Nm(I)$ (see~\cite[Prop. 4.7.4]{cohen}). Thus, we define the size of an ideal as 
$$S(I):= d^2\log(\Nm(I)).$$	
If $\ag = (1/k)I$ is a fractional ideal of $K$, where $I\subseteq\OK$ and $k\in\Z_{>0}$ is minimal, then the natural 
generalization of the notion of size is 
$$S(\ag) := \log(k) + S(I),$$
where $\log(k)$ is the base 2 logarithm of $k$.	We also define the size of elements of $K$. If $x\in\OK$ can be written as 
$x = \sum_{i\leq d}x_i\omega_i$, 
where $x_i\in\Z$, then we define its size by 
$$S(x) := d\log(\max_i|x_i|).$$
It can be generalized to elements $y\in K$ by writing $y = x/k$ where $x\in\OK$ and $k$ is a minimal positive integer, and by 
setting 
$$S(y) := \log(k) + S(x).$$
In the litterature, the size of elements of $K$ is often expressed with $\|x\|$. These two notions are in fact 
related.

\begin{proposition}\label{prop:coeff_norm}
Let $x\in\OK$, the size of $x$ and its $T_2$-norm satisfy
\begin{align*}
\log\left(\|x\|\right)&\leq \tilde{O}\left( \frac{S(x)}{d} + d^2 + \log|\Delta_K|\right)\\
S(x)&\leq \tilde{O}\left( d\left( d+ \log\left(\|x\|\right)\right)\right).
\end{align*}
\end{proposition}

\begin{proof}
In appendix
\end{proof}

So, for all $x\in \OK$,
$S(x) = O\left(\log\left( \|x\|\right)\right)$, and $\log\left( \|x\|\right) = O(S(x))$, 
where the constants are polynomial in $d$ and $\log|\Delta_K|$. 

\begin{corollary}\label{cor:S_xy}
Let $x,y\in\OK$, their size satisfies 
$$S(xy)\leq \tilde{O}\left( d^3 + d\log|\Delta_K| + S(x) + S(y)\right).$$
\end{corollary}

\section{Cost model}\label{sec:cost}
We assume that the module $M$ satifies $M\subseteq \OK^n$ and that $\OK$ is given by an 
LLL-reduced integral basis $\omega_1,\cdots,\omega_d$ such 
that $\omega_1 = 1$. The 
computation of such a basis can be done by using~\cite[Cor. 3]{Buchmannshort} 
to produce a good integral basis for $\OK$ and then reducing it with the LLL algorithm~\cite{LLL}. 
In this section, we evaluate the 
complexity of the basic operations performed during our algorithm. We rely on standard number theoretic 
algorithms. We multiply two integers of bit size $h$ in time 
$\mathcal{M}(h) \leq O\left( h \log(h)\log(\log(h))\right)$ 
using Sch\"{o}nhage-Strassen algorithm, while the addition of such integers is in $O(h)$, their division 
has complexity bounded by $O(\mathcal{M}(h))$, and the Euclidiean algorihm 
that provides their GCD has complexity $O(\log(h)\mathcal{M}(h))$ (see~\cite{moller}). 
In the following, we also refer to two standard linear algebra algorithms, namely the HNF computation over the integers 
due to Storjohann~\cite{Sto_phd} in complexity $\left(nmr^{\omega -1}\log|A|\right)^{1+o(1)}$ and Dixon's $p$-adic algorithm 
for solving linear systems in 
$$\left(n^{\omega}  \log|A| \right)^{1 + o(1)},$$ where $A\in\Z^{m\times n}$ has rank $r$ 
and has its entries bounded by $|A|$, and where $3\geq\omega\geq 2$ is the exponent of the complexity of matrix multiplication. 
We need to perform additions, multiplications and inversions of elements of $K$, as well as of fractional ideals. 
There is no reference on the complexity of these operations, although many implementations can be found. We adress 
this problem in the rest of this section. We use $\tilde{O}$ to denote the complexity were all the logarithmic 
factors are omitted.


Elements $x$ of $K$ are represented as quotients of an element of $\OK$ and a positive denominator. We add them 
naively while their multiplication is done by using a precomputed table of the $\omega_i\omega_j$ for $i,j\leq d$. 

\begin{proposition}\label{prop:op_elts}
Let $\alpha,\beta\in K$ such that $S(\alpha),S(\beta)\leq B$, then the following holds:
\begin{enumerate}
 \item $\alpha + \beta$ can be computed in $\tilde{O}(dB)$
 \item $\alpha\beta$ can be computed in $\tilde{O}\left(d^2(B + d^3 + d\log|\Delta_K|)\right)$
 \item $\frac{1}{\alpha}$ can be computed in $\tilde{O}\left(d^{\omega-1}(B + d^3 + d\log|\Delta_K|)\right)$,
\end{enumerate}
\end{proposition}

\begin{proof}
Adding $\alpha$ and $\beta$ is straightforward. Multiplying them is done by storing a precomputed 
multiplication table for the $\omega_i\omega_j$. Finally, inverting $\alpha$ boils down to solving 
a linear system in the coefficients of $\frac{1}{\alpha}$. More details are given in appendix.
\end{proof}


Ideals of $\OK$ are given by their HNF representation with respect to the integral basis $\omega_1,\cdots,\omega_d$ 
of $\OK$. It consists of the HNF of the matrix representing the $d$ generators of their $\Z$ basis as rows. Operations on this 
matrix yield the addition, multiplication and inverse of an integral ideal. The corresponding operations on fractional 
ideals are trivialy deduced by taking care of the denominator.

\begin{proposition}\label{prop:op_ideals}
Let $\ag$ and $\bg$ be fractional ideals of $K$ such that $S(\ag),S(\bg)\leq B$, then the following holds:
\begin{enumerate}
 \item $\ag + \bg$ can be computed in $\tilde{O}(d^{\omega+1}B)$,
 \item $\ag\bg$ can be computed in $\tilde{O}(d^3(d^4 + d^2\log|\Delta_K| + B))$,
 \item $1/\ag$ can be computed in $\tilde{O}\left( d^{2\omega}(d^4 + d^2\log|\Delta_K| + B)\right)$.
\end{enumerate}
\end{proposition}

\begin{proof}
The addition of integral ideals $\ag$ and $\bg$  given by their HNF matrix $A$ and $B$ is given by the HNF of $\left(\frac{A}{B}\right)$. 
To multiply them, one has to compute the HNF of the matrix whose $d^2$ rows represent $\gamma_i\delta_j$ where 
$(\gamma_i)_{i\leq d}$ is an integral basis for $\ag$ and $(\delta_i)_{i\leq d}$ is an integral basis for $\bg$. Finally, 
following the approach of~\cite[4.8.4]{cohen}, inverting $\ag$ boils down to solving a $d^2\times (d+d^2)$ linear 
system. More details are given in appendix.
\end{proof}

Note that the reason why the dependency in $B$ in the complexity of the addition of fractional ideals is 
slightly more than in the complexity of the multiplication is the way 
we deal with the denominators. In the case of integral ideals, the addition would be in $\tilde{O}(d^{\omega-1}B)$. 
The last operation that needs to be performed during our HNF algorithm is the multiplication between an element of 
$K$ and a fractional ideal.

\begin{proposition}\label{prop:mult_id_elt}
Let $\alpha\in K$, a fractional ideal $\ag\subseteq K$ and $B_1,B_2$ such that 
$S(\ag)\leq B_1$ and $S(\alpha)\leq B_2$, then $\alpha\ag$ can be computed in 
expected time bounded by 
$$\tilde{O}\left(d^{\omega}\left(d^3 + d\log|\Delta_K| +  \frac{B_1}{d} + B_2 \right)\right).$$
\end{proposition}

\begin{proof}
If $\gamma_1,\cdots,\gamma_d$ is an integral basis for $\ag\subseteq \OK$, then $(\alpha\gamma_i)_{i\leq d}$ is 
one for $(\alpha)\ag$. The HNF of the matrix representing these elements leads to the desired result. More details 
are given in appendix.
\end{proof}

\section{The normalization}

The normalization is the key difference between our approach and the one of Cohen~\cite[1.5]{cohen2}. It is the 
strategy that prevents the coefficient swell by calculating a pseudo-basis for which the ideals are integral 
with size bounded by the field's invariants. Given a one-dimensional $\OK$-module $\ag A\subseteq\OK^n$ where 
$\ag$ is a fractional ideal of $K$, and $A\in K^n$, we find $b\in K$ such that the size taken to represent our 
module as $(b\ag)(A/b)$ is reasonably bounded. Indeed, any non trivial module can be represented by elements of 
arbitrary large size, which would cause a significant slow-down in our algorithm.

The first step to our normalization is to make sure that $\ag$ is integral. This allows us to bound the denominator 
of the coefficients of the matrix when manipulating its rows during the HNF algorithm. If $k\in\Z$ is the denominator 
of $\ag$,
then replacing $\ag$ by $k\ag$ and $A$ by $A/k$ increases the size needed to represent our module via the 
growth of all the denominators of the coefficients of $A\in K^n$. Thus, after this operation, 
the size of each coefficient $a_i$ of $A$ is bounded by $S(a_i) + S(\ag)$.

We can now assume that our one-dimensional module is of the form $\ag A$ where $\ag\subseteq\OK$ and 
$A\in K^n$ at the price of a slight growth of its size. The next step of normalization is to express our 
module as $\ag' A'$ where $A'\in K^n$ and $\ag'\subseteq\OK$ such that $\Nm(\ag')$ only depends on invariants of 
the field. To do this, we invert $\ag$ and write it as 
$$\ag^{-1} = \frac{1}{k}\bg,$$
where $k\in\Z_{>0}$ and $\bg\subseteq\OK$. As $\Nm(\ag)\in\ag$, we have $\Nm(\ag)\ag^{-1}\subseteq\OK$ and 
thus $k\leq \Nm(\ag)$. Therefore, 
$$\Nm(\bg) \leq \frac{\Nm(k)}{\Nm(\ag)}\leq \frac{k^{d}}{\Nm(\ag)} \leq \Nm(\ag)^{d-1}.$$
Then we use the LLL algorithm to find an element $\alpha\in\bg$ such that
$$\|\alpha\| \leq d^{1/2}2^{d/2}|\Delta_K|^{1/2d}\Nm(\bg)^{1/d}.$$
Our reduced ideal is 
$$\ag' := \left(\frac{\alpha}{k}\right)\ag \subseteq \ag^{-1}\ag = \OK.$$
The integrality of $\ag'$ comes from the definition of $\bg^{-1}$ and the fact that $\alpha\in\bg$. From the 
arithmetic-geometric mean, we know that $\Nm(\alpha)\leq \frac{\|\alpha\|^d}{d^d}$, therefore 
$$\Nm(\alpha)\leq 2^{d^2/2}\sqrt{|\Delta_K|}\Nm(\bg),$$
and the norm of the reduced ideal can be bounded by $\Nm(\ag')\leq 2^{d^2/2}\sqrt{|\Delta_K|}$. On the other 
hand, we set $A' := (k/\alpha)A$, which induces a growth of the coefficients $a_i$ of $A$. Indeed, each 
$a_i$ is multiplied by $(k/\alpha)$. 

\begin{proposition}\label{prop:size_normalization}
The size of the normalized module $\ag'A'$ of $\ag A\subseteq K^n$ satisfies 
\begin{align*}
S(a'_i) &\leq \tilde{O}\left(  d^3 + d\log|\Delta_K| + S(\ag) + S(a_i)\right)\\
S(\ag') &\leq \tilde{O}\left( d^3 + d\log|\Delta_K|\right)
\end{align*}
\end{proposition}

\begin{proof}
From Corollary~\ref{cor:S_xy} we know that 
$$S\left(\frac{a_ik}{\alpha}\right) \leq \tilde{O}\left( d^3 + d\log|\Delta_K| + \frac{S(\ag)}{d} + S(a_i) + S\left(\frac{1}{\alpha}\right)\right)$$

In addition, if $\frac{1}{\alpha} = \frac{x}{k'}$ where $x\in\OK$ and $k'\in\Z_{>0}$, then 
$$S\left(\frac{1}{\alpha}\right)\leq \tilde{O}\left( \log(k') + d(d+\log\|x\|)\right).$$
On the one hand, we have 
\begin{align*}
k' \leq \Nm(\alpha)\leq 2^{d^2/2}\sqrt{|\Delta_K|}\Nm(\ag)^{d-1},
\end{align*}
and on the other hand, we need to bound $\|x\|$. We notice that since $\Nm(\alpha)\in\Q$, 
$\forall j\leq d$, $\Nm(\alpha) = \alpha\beta = \sigma_j(\alpha\beta)$. We also know that 
$\forall j$, $|\sigma_j(\alpha)|\leq \|\alpha\|$. Therefore, 
$$\forall j\leq d, \ |\sigma_j(x)| = \frac{\Nm(\alpha)}{|\sigma_j(\alpha)|} 
 = \prod_{i\neq j}|\sigma_i(\alpha)|\leq \|\alpha\|^{d-1}.$$
Therefore $\|x\|\leq \sqrt{d}\|\alpha\|^{d-1}$, and thus
$$S\left(\frac{1}{\alpha}\right)\leq \tilde{O}\left(  d^3 + d\log|\Delta_K| + S(\ag)\right).$$
\end{proof}

Our normalization, summarized in Algorithm~\ref{alg:normalization}, was performed at the price of a 
reasonable growth in the size of the object we manipulate. Let us now evaluate its complexity.
\begin{algorithm}[ht]
\caption{Normalization of a one-dimensional module}
\begin{algorithmic}[1]\label{alg:normalization}
\REQUIRE $A \in K^{n}$, fractional ideal $\ag$ of $K$.
\ENSURE  $A' \in K^{n}$, $\ag'\subseteq \OK$ such that $\Nm(\ag')\leq 2^{d^2/2}\sqrt{|\Delta_K|}$ and 
$\ag A = \ag' A'$.
\STATE $\ag \leftarrow k_0\ag$, $A \leftarrow A/k_0$ where $k_0$ is the denominator of $\ag$.
\STATE $\bg \leftarrow k \ag^{-1}$ where $k$ is the denominator of $\ag^{-1}$.
\STATE Let $\alpha$ be the first element of an LLL-reduced basis of $\bg$.
\STATE $\ag'\leftarrow \left(\frac{\alpha}{k}\right) \ag$, $A' \leftarrow \left(\frac{k}{\alpha}\right)A$.
\RETURN $\ag'$, $A'$.
\end{algorithmic}
\end{algorithm}

\begin{proposition}\label{prop:normalization}
Let $B_1,B_2$ such that $S(\ag)\leq B_1$ and $\forall i, S(a_i)\leq B_2$, then the 
complexity of Algorithm~\ref{alg:normalization} is bounded by 
$$\tilde{O}\left(nd^2(d^3 + B_1 + B_2  + d\log|\Delta_K|)\right).$$
\end{proposition}

\begin{proof}
The inversion of $\ag$ is performed in time 
$$\tilde{O}\left( d^{2\omega}(d^4 + d^2\log|\Delta_K| + B_1)\right),$$
by using Proposition~\ref{prop:op_ideals}. Then, 
the LLL-reduction of the basis of $\bg$ is done by the $L^2$ algorithm of Stehl\'{e} and Nguyen~\cite{NguSte09b} 
in expected time bounded by 
$$\tilde{O}\left(d^3\left(d + \frac{S(\bg)}{d^2}\right)\frac{S(\bg)}{d^2}d\right)\leq \tilde{O}\left(d^2S(\ag)(d^2+S(\ag))\right).$$
Then, computing $(\alpha/k)\ag$ is the multiplication of the ideal $\ag$ by the element 
$\alpha/k$ which satisfies 
$$S(\alpha/k)\leq \tilde{O}\left( d^2 + \log|\Delta_K| + S(\ag)/d\right).$$ 
This takes $\tilde{O}\left(d^{\omega-1}(S(\ag) + d^4 + d^2\log|\Delta_K|)\right)$. Finally, computing 
$k(1/\alpha)A$ consists of inverting $\alpha$ with $S(\alpha)\leq \tilde{O}(d^3 + \log|\Delta_K| + B_1/d)$, which takes
$$\tilde{O}\left(d^{\omega-1}(d^3 + B_1/d + d\log|\Delta_K|)\right),$$

and performing $n+1$ multiplications between elements of size bounded by $\tilde{O}(d^3 + B_1 + B_2 + d\log|\Delta_K|)$, which 
is done in time 
$$\tilde{O}\left(nd^2(d^3 + B_1 + B_2  + d\log|\Delta_K|)\right).$$
The result follows from the combination of the above expected times and from the fact that $2\leq \omega\leq 3$. 
\end{proof}

\section{Reduction modulo a fractional ideal}

To achieve a polynomial complexity for our HNF algorithm, we reduce some elements of $K$ modulo 
ideals whose norm can be reasonably bounded. We show in this section how to bound the norm of a reduced 
element with respect to the norm of the ideal and invariants of $K$. Let $\ag$ be a fractional ideal of 
$K$, and $x\in K$. Our goal is to find $\overline{x} \in K$ such that 
$\| \overline{x} \|$ 
is bounded, and that $x - \overline{x} \in \ag$. 

The reduction algorithm consists of finding an LLL-reduced basis $r_1,\cdots , r_d$ of $\ag$ and to decompose 
$$x = \sum_{i\leq d } x_i r_i.$$
Then, we define 
$$\overline{x} := x - \sum_{i\leq d } \lfloor x_i \rceil r_i.$$

\begin{proposition}
Let $x\in K$ and $\ag$ be a fractional ideal of $K$, then Algorithm~\ref{alg:reduction} returns $\overline{x}$ such that 
$x - \overline{x}\in\ag$ and 
$$\|\overline{x}\| \leq d^{3/2}2^{d/2}\Nm(\ag)^{1/d}\sqrt{|\Delta_K|}.$$
\end{proposition}

\begin{proof}
In appendix
\end{proof}

\begin{algorithm}[ht]
\caption{Reduction modulo a fractional ideal}
\begin{algorithmic}[1]\label{alg:reduction}
\REQUIRE $\alpha\in K$, fractional ideal $\ag$ of $K$.
\ENSURE  $\overline{\alpha}\in K$ such that $\alpha - \overline{\alpha} \in \ag$ and $\|\overline{\alpha}\| \leq d^{3/2}2^{d/2}\Nm(\ag)^{1/d}\sqrt{|\Delta_K|}$.
\IF{$\|\alpha\| \leq d^{3/2}2^{d/2}\Nm(\ag)^{1/d}\sqrt{|\Delta_K|}$ \textbf{or} $\alpha = 1$}
\RETURN $\alpha$.
\ELSE
\STATE Compute an LLL-reduced basis $(r_i)_{i\leq d}$ of $\ag$.
\STATE Decompose $\alpha = \sum_{i\leq d} x_i r_i$.
\STATE $\overline{\alpha} \leftarrow \alpha - \sum_{i\leq d } \lfloor x_i \rceil r_i.$
\RETURN $\overline{\alpha}$.
\ENDIF
\end{algorithmic}
\end{algorithm}

\begin{proposition}\label{prop:red_LLL}
Let $B_1,B_2$ such that $S(\ag)\leq B_1$ and $S(\alpha)\leq B_2$, then the complexity of Algorithm~\ref{alg:reduction} 
is bounded by 
$$\tilde{O}\left( B_1(d^3 + B_1) + d^{\omega - 1}B_2 + d^{\omega+2} \right)$$
\end{proposition}

\begin{proof}
To compute the LLL-reduced basis of $\ag$, we LLL-reduce the integral ideal $k\ag$ where $k\in\Z_{>0}$ is 
the denominator of $\ag$. Then, we express $x$ with respect to the basis of $k\ag$ where $x\in\OK$ satifies 
$\alpha = x/a$ for $a\in\Z_{>0}$. Then we divide by the respective denominator at the extra cost of $d$ 
multiplications.

Using the $L^2$ algorithm of Stehl\'{e} and Nguyen~\cite{NguSte09b} yields the reduced basis of $k\ag$ 
in expected time bounded by 
$$\tilde{O}\left( d^3\left( d + \frac{S(\ag)}{d^2}\right) \frac{S(\ag)}{d^2}d\right)\leq \tilde{O}\left( S(\ag)(d^3 + S(\ag))\right).$$
Then, expressing $x$ with respect to the reduced basis of $k\ag$ costs 
$$\tilde{O}\left( d^{\omega}\left(\frac{S(\ag)}{d^2} + d\log|\Delta_K| +d^2 + \frac{S(x)}{d}\right)\right).$$
Finally, the subtraction and the division by the denominators are in
$$\tilde{O}\left(d\frac{S(\alpha)}{d}\right).$$
\end{proof}

\section{Modular HNF Algorithm}

Let $M \subseteq \OK^n$ be an $\OK$-module. 
We use a variant of the modular version of~\cite[Alg. 1.4.7]{cohen2} which ensures that the current 
pseudo-basis $[\ag_i,A_i]_{i\leq n}$ of the module satisfies $\ag\subseteq\OK$ at every step of the 
algorithm. This extra feature allows us to bound the denominator of coefficients of the 
matrix whose rows we manipulate. Algorithm~\ref{alg:HNF} computes the HNF modulo the determinantal ideal $\g$, and 
Algorithm~\ref{alg:Euclidian} recovers an actual HNF for $M$. In this section, we discuss the differences 
between Algorithms~\ref{alg:HNF} and~\ref{alg:Euclidian} and their equivalent in~\cite[1.4]{cohen2}. 

After the original normalization, all the ideals are integral. As $M\subseteq\OK^n$, we 
immediatly deduce that the ideal $\dg$ created at Step~6 
of Algorithm~\ref{alg:HNF} is integral as well. In addition, from the definition of the inverse of an ideal 
we also have that 
$$\frac{b_{i,i}\bg_ib_{i,j}\bg_j}{b_{i,j}\bg_j + b_{i,i}\bg_i}\subseteq \OK,$$
which allows us to conclude that the update of $(\bg_i,\bg_j)$ performed at Step~9 of Algorithm~\ref{alg:HNF}
preserves the fact that our ideals are integral. 


\begin{algorithm}[ht]
\caption{HNF of a full-rank square pseudo-matrix modulo $\g$}
\begin{algorithmic}[1]\label{alg:HNF}
\REQUIRE $A \in K^{n\times n}$, $\ag_1,\cdots,\ag_n$ , $\g$.
\ENSURE pseudo-HNF $B$, $\bg_1,\cdots,\bg_n$ modulo $\g$. 
\STATE $B\leftarrow A$, $\bg_i\leftarrow \ag_i$, $j\leftarrow n$.
\STATE Normalize $[(B_i),(\bg_i)]_{i\leq n}$ with Algorithm~\ref{alg:normalization}
\WHILE { $j\geq 1$ }
\STATE $i \leftarrow j-1$.
\WHILE { $i\geq 1$ } 
\STATE $\dg \leftarrow b_{i,j}\bg_i + b_{j,j}\bg_j$
\STATE Find $u\in \bg_i\dg^{-1}$ and $v\in \bg_j\dg^{-1}$ such that $b_{i,j}u+ b_{j,j}v = 1$ 
with~\cite[Th. 1.3.3]{cohen2}.
\STATE $(B_i,B_j)\leftarrow (b_{j,j}B_i-b_{i,j}B_j,uB_i + vB_j)$.
\STATE $(\bg_i,\bg_j)\leftarrow (b_{i,j}\bg_ib_{j,j}\bg_j\dg^{-1},\dg)$.
\STATE Normalize $\bg_i,B_i$ with Algorithm~\ref{alg:normalization}.
\STATE Reduce $B_i$ modulo $\g\bg_i^{-1}$ and $B_j$ modulo $\g\bg_j^{-1}$ with Algorithm~\ref{alg:reduction}. 
\STATE $i\leftarrow i-1$.
\ENDWHILE
\STATE $j\leftarrow j-1$.
\ENDWHILE
\RETURN $(\bg_i)_{i\leq n}$, $B$.
\end{algorithmic}
\end{algorithm}

The normalization and reduction at Step 10-11 allow us to keep the size of the $B_i$ and of the $\bg_i$ 
reasonably bounded by invariants of $K$ and the dimension of the module. By doing so, we give away some 
information about the module $M$. However, algorithm~\ref{alg:Euclidian} allows us to recover $M$, as we 
state in Proposition~\ref{prop:validity}.

\begin{proposition}\label{prop:validity}
The $\OK$-module defined by the pseudo-basis $[(W_i),(\cg_i)]$ obtained by applying Algorithm~\ref{alg:Euclidian} to 
the HNF of $M$ modulo $\g(M)$ satisfies 
$$\cg_1 W_1 + \cdots + \cg_n W_n = M.$$
\end{proposition} 

\begin{proof}
The proof of this statement essentially follows its equivalent for matrices over the integers. 
It consists of showing that $W:= \sum_i\cg_i$ and $M:= \sum_i\ag_iA_i$ have the same determinantal 
ideal and that $W\subseteq A$, and then showing that this implies that $W = M$. A more complete proof is given 
in appendix.  
\end{proof}

\begin{algorithm}[ht]
\caption{Eucledian reconstruction of the HNF}
\begin{algorithmic}[1]\label{alg:Euclidian}
\REQUIRE $B \in K^{n\times n}$, $\bg_1,\cdots,\bg_n$ output of Algorithm~\ref{alg:HNF} modulo $\g$ for 
$M\subseteq \OK^n$.
\ENSURE An HNF $W$,$\cg_1,\cdots,\cg_n$ for $M$. 
\STATE $j\leftarrow n$ , $\g_j\leftarrow \g$.
\WHILE { $j\geq 1$ }
\STATE $\cg_j\leftarrow \bg_j + \g_j$.
\STATE Find $u\in \bg_j\dg^{-1}$ and $v\in \g\cg^{-1}_j$ such that $u + v = 1$.
\STATE $W_j\leftarrow uB_j\bmod \g\cg^{-1}_j$.
\STATE $\g_j \leftarrow \g_j\cg^{-1}_j$.
\STATE $j\leftarrow j-1$.
\ENDWHILE
\RETURN $W,(\cg_i)_{i\leq n}$.
\end{algorithmic}
\end{algorithm}

\section{Complexity of the modular HNF}\label{sec:complexity}

Let us assume that we are able to compute the determinantal ideal $\g$ of our module $M$ in polynomial time 
with respect to the bit size of the invariants of the field and of $S(\g)$. We discuss the computation of $\g$ 
in Section~\ref{sec:mod_comp}. In this section, we show that 
Algorithm~\ref{alg:HNF} and Algorithm~\ref{alg:Euclidian} are polynomial wih respect to the same parameters.
This result is analogous to the case of integers matrices. Indeed, the only thing we need to verify is that 
the size of the elements remains reasonably bounded during the algorithm.

In Algorithm~\ref{alg:HNF}, the coefficient explosion is prevented by the modular reduction of Step~11. It ensures that 
$$\forall i_1, i_2<j, \ \|b_{i_1,i_2}\|\leq d^{3/2}2^{d/2}\Nm(\g\bg_{i_1}^{-1})^{1/d}\sqrt{|\Delta_K|}.$$
This is not enough to prevent the explosion since $b_{i_1,i_2}$ might not be integral. Therefore, there is a 
minimal $k\in\Z_{>0}$ such that $kb_{i_1,i_2}\in\OK$, which we need to bound to ensure that $S(b_{i_1,i_2})$ remains bounded as well. We know that $b_{i,j}\bg_i\subseteq\OK$, and that $\bg_i$ is integral. Thus, 
$\Nm(k)\mid \Nm(\bg_{i_1})$, which in turns implies that $k\leq \Nm(\bg_{i_1})$. As on the other hand, the 
normalization of Step~10 ensures that $\Nm(\bg_{i_1})\leq 2^{d^2/2}\sqrt{|\Delta_K|}$, we conclude that 
after Step~11, 
\begin{align*}
S(b_{i_1,i_2})\leq  \tilde{O}\left( d^2 + d\log|\Delta_K| + \frac{S(\g)}{d^2}\right).
\end{align*}

In Algorithm~\ref{alg:HNF}, we last manipulate $B_j$ and $\bg_j$ when the index $j$ is the pivot. In that 
case, we cannot use the normalization to bound the size since we require that $b_{j,j} = 1$. However we 
reduce $B_j$ modulo $\g\bg_j$, which means that 
$$\forall i\leq j,\ \|b_{i,j}\|\leq d^{3/2}2^{d/2}\Nm(\g\bg_i^{-1})^{1/d}\sqrt{|\Delta_K|}.$$
In addition, the arithmetic-geometric tells us that $\|b_{j,j}\|\geq \sqrt{d}\Nm(b_{i,j})^{1/d}$, which in turn implies that 
\begin{equation}\label{eq:bound1}
\forall i\leq j,\ \Nm(b_{i,j}\bg_i)\leq d^d 2^{d^2/2} \Nm(\g)^d |\Delta_K|^{d/2}.
\end{equation}
As we know that 
$$\Nm(b_{i,j}\bg_i + b_{j,j}\bg_j) \leq \max\left(\Nm(b_{i,j}\bg_i),\Nm(b_{j,j}\bg_j)\right),$$
 we therefore know that after Step~9 
$$\Nm(\bg_j)\leq d^d 2^{d^2/2} \Nm(\g)^d |\Delta_K|^{d/2},$$ which allows us to bound the size of the denominators in the $j$-th row 
the same way we did for the rows of index $i_1<j$:
$$\forall i\leq j, \ S(b_{i,j})\leq  \tilde{O}\left( d^2 + d\log|\Delta_K| + \frac{S(\g)}{d^2}\right).$$
\begin{proposition}
The complexity of Algorithm~\ref{alg:HNF} is in 
$$\tilde{O}\left( n^3d^2\left(d^3 + d^2\log|\Delta_K| + S(\g)\right)^2\right).$$
\end{proposition}

\begin{proof}
Steps~6 to~11 of Algorithm~\ref{alg:HNF} are repeated $O(n^2)$ times. Let us analyze their complexity. First, at Step~6 
we have 
\begin{align}
S(b_{i,j}) &\leq \tilde{O}\left(d^3 + \log|\Delta_K| + \frac{S(\g)}{d^2}\right)\label{eq:bound_bij}\\
S(\bg_i) &\leq \tilde{O}(d^3 + \log|\Delta_K|)
\end{align}
so from Proposition~\ref{prop:mult_id_elt}, computing $b_{i,j}\bg_i$ takes 
$$\tilde{O}\left( d^{\omega-2}(d^5 + d^{3}\log|\Delta_K| + S(\g))\right).$$
Then, from Proposition~\ref{prop:op_ideals} and~\eqref{eq:bound1}, 
$$S(b_{i,j}\bg_i)\leq \tilde{O}\left(d^4 + dS(\g) + d^3\log|\Delta_K|\right),$$ 
and computing $\dg$ 
costs
$$\tilde{O}\left( d^{\omega + 2}(d^3 + d^2\log|\Delta_K| + S(\g))\right).$$
As $S(\dg)\leq S(b_{i,j}\bg_i)$, computing $\dg^{-1}$ takes
$$\tilde{O}\left( d^{2\omega+1}\left( d^3 + d^2\log|\Delta_K| + S(\g) \right)\right).$$
From~\cite[4.8.4]{cohen}, this is done by solving a linear system on 
a matrix $D$ satisfying 
$$\log|D|\leq \tilde{O}\left( d^2 + \log|\Delta_K| + \frac{S(\g)}{d^2}\right),$$
and the coefficients of the HNF matrix of $\dg$ are those of a matrix $M$ satisfying 
$\log|\det(M)|\leq \tilde{O}\left(d^2\log|D|\right)$. Therefore, we have 
$$S(\dg^{-1})\leq d^2\log|\det(M)| \leq \tilde{O}\left( d^2\left( d^4 + d^2\log|\Delta_K| + S(\g)\right)\right).$$
As $S(\bg_i),S(\bg_j)\leq \tilde{O}(d^3 + \log|\Delta_K|)$, computing $\bg_i\dg^{-1}$ and 
$\bg_j\dg^{-1}$ takes 
$$\tilde{O}\left( d^{5}\left( d^4 + d^2\log|\Delta_K| + S(\g)\right)\right).$$
Then, from~\cite[Th. 1.3.3]{cohen2}, computing $u$ and $v$ is done by finding $u'\in b_{i,j}\bg_i\dg^{-1}$ and 
$v'\in b_{j,j}\bg_j\dg^{-1}$ such that $u'+v'=1$ and returning $u:= u'/b_{i,j}$ and $v:= v'/b_{j,j}$. Let 
$I_i := b_{i,j}\bg_i\dg^{-1}\subseteq \OK$. Then, from~\cite[Prop. 1.3.1]{cohen2} computing $u',v'$ is done at the 
cost of an HNF computation of a $2d\times d$ matrix whose entries have their size bounded by $\log(\Nm(I_j))$. This 
cost is in 
$$\tilde{O}\left( d^{\omega}(d^3 + d^2\log|\Delta_K| + S(\g))\right).$$
In addition, $S(u'),S(v')\leq \tilde{O}(d^4 +d^3\log|\Delta_K| + dS(\g))$. 
Then, by using the same methdods as in 
the proof of Proposition~\ref{prop:size_normalization}, we know that 
$S\left(\frac{1}{b_{i,j}}\right)\leq \tilde{O}\left( d^3 + \frac{S(\g)}{d} + d^2\log|\Delta_K|\right)$ while Proposition~\ref{prop:op_elts} 
ensures us that inverting $b_{i,j}$ is done in 
$$\tilde{O}\left( d^{\omega-1}\left( d^3 + d\log|\Delta_K| + \frac{S(\g)}{d^2}\right) \right).$$
Then, calculating $u'/b_{i,j}$ and $v'/b_{j,j}$ is done in time bounded by
$$\tilde{O}\left( d^2 \left( d^4 + d^3\log|\Delta| + dS(\g)\right) \right),$$
and by Corollary~\ref{cor:S_xy}, we know that 
$$S(u),S(v)\leq \tilde{O}\left(d^4 + d^3\log|\Delta_K| + dS(\g)\right).$$
Then, from Proposition~\ref{prop:op_elts} and~\eqref{eq:bound_bij}, the expected time for Step~8 is 
bounded by 
$$\tilde{O}\left( nd^2(d^4 + d^3\log|\Delta_K| + dS(\g))\right).$$
In addition, after Step~8, we have 
$$S(b_{i,j})\leq  \tilde{O}\left( d^4 + d^3\log|\Delta_K| + dS(\g)\right).$$
Then, from Proposition~\ref{prop:op_ideals} and the bounds on $S(b_{i,j}\bg_i)$ and $S(\dg^{-1})$ computed 
above, Step~9 
takes 
$$\tilde{O}\left( d^{5}\left( d^4 + d^2\log|\Delta_K| + S(\g)\right)\right).$$
By using Proposition~\ref{prop:normalization}, we bound the time taken by Step~10 by 
$$\tilde{O}\left( nd^3(d^3 + d^2\log|\Delta_K| + S(\g))\right),$$
Finally, from the 
bound on $S(b_{i,j})$ after Step~8 and Proposition~\ref{prop:red_LLL}, Step~11 takes 
$$\tilde{O}\left( nd^2\left(d^3 + d^2\log|\Delta_K| + S(\g)\right)^2\right).$$
\end{proof}

The Euclidian reconstruction of Algorithm~\ref{alg:Euclidian} can be seen as another pivot operation 
between the two one-dimensional $\OK$-modules $\bg_j B_j$ and $\g_je_j$ for each $j\leq n$. We can 
therefore bound the entries of $W$ by the same method as for Step~6-11 of Algorithm~\ref{alg:HNF}, we the 
extra observation 
$$\Nm(\g_j)\leq \Nm(\g).$$
Therefore, we showed that we could bound the size of the objects that are manipulated 
throughout the algorithm by values that are polynomial in terms of $n$, $d$, $S(\g)$ and $\log(|\Delta_K|)$, 
and that the complexity of 
the HNF algorithm was polynomial in these parameters.

\section{Computing the modulus}\label{sec:mod_comp}

Let us assume that $A\in\OK^{n\times n}$. If it is not the case, then we need to multiply by the common 
denominator $k$ of the entries of $A$ and return $\det(kA)/k^n$. In this section, we describe how to compute $\g$ in polynomial time with respect to $n$, $d$, $\log|\Delta_K|$ and the size of the entries of $A$. The idea is to compute $\det(A)\mod (p)$ for a sufficiently large prime number $p$. In practice, one might 
prefer to compute $\det(A)\mod (p_i)$ for several prime numbers $p_1,\cdots,p_l$ and recombine the values 
via the chinese remainder theorem, but for the sake of simplicity, we only describe that procedure for a 
single prime. Once $\det(A)$ is computed in polynomial time, we return 
$$\g = \det(A)\cdot \ag_1\cdots \ag_n.$$

The first step consists of evaluating how large $p$ should be to ensure that we recover $\det(A)$ 
uniquely. As $p\omega_1,\cdots,p\omega_d$ is an integral basis for $(p)$, it suffices that 
$p\geq \max_i|a_i|$ where $\det(A) = \sum_i a_i \omega_i$. As $\max_i|a_i|\leq 2^{3d/2}\|\det(A)\|$, it 
suffices to bound $\|\det(A)\|$. 
We first compute an upper bound on $|\sigma(\det(A))|$ for the $d$ complex embeddings $\sigma$ of $K$ via Hadamard's inequality and then we deduce a bound on $\|\det(A)\|$. Let $\sigma : K\rightarrow\mathcal{C}$, we know from Hadamard's inequality that 
$$|\sigma(\det(A))|\leq B^n n^{n/2},$$
where $B$ is a bound on $\sigma(a_{i,j})$. Such a bound can be derived from the size of the coefficient 
of $A$ by using 
$$\forall x, \ \forall i\ |\sigma_i(x)| \leq \left(\max_j|x_j|\right)d^{3/2}2^{d^2/2}\sqrt{|\Delta_K|}.$$
This way, we see that $B := 2^{\max_{i,j}\left(S(a_{i,j})\right)}d^{3/2}2^{d^2/2}\sqrt{|\Delta_K|}$ suffices. Then, our 
bound on $\|\det(A)\|$ is simply 
$$\|\det(A)\| \leq \sqrt{n} 2^{\max_{i,j}\left(S(a_{i,j})\right)}d^{3/2}2^{d^2/2}\sqrt{|\Delta_K|}.$$

\begin{algorithm}[ht]
\caption{Computation of $\det(A)$}
\begin{algorithmic}[1]\label{alg:det_ideal}
\REQUIRE $A\in \OK^{n\times n}$, $B > \max_{i,j}\left(S(a_{i,j})\right)$
\ENSURE $\det(A)$.
\STATE Let $p\geq \sqrt{n} 2^{B}d^{3/2}2^{d^2/2}\sqrt{|\Delta_K|}$ be a prime.
\FOR {$\p_i\mid(p)$}
\STATE Compute $\det(A)\bmod \p_i$.
\ENDFOR
\STATE Recover $\det(A)\bmod (p)$ via successive applications of Algorithm~\ref{alg:CRT}
\RETURN $\det(A)$.
\end{algorithmic}
\end{algorithm}

To reconstruct $\det(A)\mod (p)$ from $\det(A)\bmod\ag_i$ for $i\leq d$, let us consider 
the simpler case of the reconstruction modulo two coprime ideals $\ag,\bg$ of $\OK$. Let 
$M_\ag$ and $M_\bg$ be the matrices representing the $\Z$ basis of $\ag$ and $\bg$ in the 
integral basis $(\omega_i)_{i\leq d}$ of $\OK$, and let $x,y,w\in\OK$ such that 
\begin{align*}
x &= y\bmod \ag \\
x &= w\bmod \bg.
\end{align*}
We wish to compute $z\in\OK$ such that $x = z\bmod \ag\bg$. As in~\cite[Prop. 1.3.1]{cohen2}, we 
can derive $a\in\ag,b\in\bg$ such that $a + b = 1$ from the HNF of $\left(\frac{M_\ag}{M_\bg}\right)$. 
Then, a solution to our CRT recomposition is given by 
$$z := wa + yb.$$

\begin{algorithm}[ht]
\caption{CRT recomposition}
\begin{algorithmic}[1]\label{alg:CRT}
\REQUIRE $\ag,\bg\subseteq\OK$, $x,y,w\in\OK$ such that $x = y\bmod\ag$ and $x = w \bmod\bg$.
\ENSURE $z\in\OK$ such that $x = z\bmod\ag\bg$.
\STATE Compute $a\in\ag,b\in\bg$ such that $a + b = 1$. 
\RETURN $z$.
\end{algorithmic}
\end{algorithm}

\begin{proposition}
Let $B>\max_{i,j}\left(S(a_{i,j})\right)$, then the complexity of Algorithm~\ref{alg:det_ideal} is bounded by 
$$\tilde{O}\left( n^3d^7(d^2 + B + \log|\Delta_K|)^2\right).$$
\end{proposition}

\begin{proof}
For each $\p_i$, the computation of $\det(A)\bmod\p_i$ consists of $n^3$ multiplications of reduced 
elements modulo $\p_i$ followed by a reduction modulo $\p_i$. Given our choice of $p$, we have 
$$\log\Nm(\p_i)\leq \tilde{O}\left( d(B + d^2 + \log|\Delta_K|)\right).$$
Therefore, the size of the elements $x\in\OK$ involved in these multiplications satisfies 
$$S(x)\leq \tilde{O}\left( d^2(d^2 + \log|\Delta_K| + B)\right).$$
The cost of the multiplications is in 
$$\tilde{O}\left( d^4(d^2 + B + \log|\Delta_K|)\right),$$
while the mdular reductions cost 
$$\tilde{O}\left( d^6(d^2 + B + \log|\Delta_K|)^2\right).$$
The time to reconstruct $\det(A)\bmod(p)$ corresponds to the computation of $n^2$ Hermite forms of 
$d^2\times d$ integer matrices $M$ such that $\log|M|\leq \log(\Nm(\p_i))$. This takes 
$$\tilde{O}\left( n^2d^{\omega+3}(d^2 + B + \log|\Delta_K|)^2\right).$$
\end{proof}

\section{Conclusion}

We described a polynomial time algorithm for computing the HNF basis of an $\OK$-module. Our strategy 
relies on the one of Cohen~\cite[1.4]{cohen2} who had conjectured that his modular algorithm 
was polynomial. The crucial difference between our algorithm and the one of~\cite[1.4]{cohen2} 
is the normalization which allows us to prove the complexity to be polynomial. Without it, we cannot bound 
the denominator of the coefficients of the matrix when 
we recombine rows, even if they are reduced modulo the determinantal ideal. We provided a rigorous proof 
of the complexity of our method with respect to the dimension of the module, the size of the input and the 
invariants of the field. Our algorithm is the first polynomial time method for computing the HNF of an $\OK$-module. 
This result is significant since other applications rely on the 
possibility of computing the HNF of an $\OK$-module in polynomial time. In particular, Fieker and Stehl{\'e}~\cite{stehle_fieker_LLL} made this assumption in the analysis of their LLL algorithms 
for $\OK$-modules. Our result has natural ramifications in cryptography through the LLL algorithm of Fieker and 
Stehl{\'e}~\cite{stehle_fieker_LLL}, but it can also be used for 
list-decoding number field codes.
%
\bibliographystyle{abbrv}
\bibliography{paper}  
%
%
\pagebreak
\appendix

\section{Detailed proofs of statements}

\subsection{notion of size}

\setcounter{proposition}{0}
\begin{proposition}
Let $x\in\OK$, the size of $x$ and its $T_2$-norm satisfy
\begin{align*}
\log\left(\|x\|\right)&\leq \tilde{O}\left( \frac{S(x)}{d} + d^2 + \log|\Delta_K|\right)\\
S(x)&\leq \tilde{O}\left( d\left( d+ \log\left(\|x\|\right)\right)\right).
\end{align*}
\end{proposition}

\begin{proof}
Let us show how $S(x)$ and $\log(\|x\|)$ are 
related. First, we can assume~\cite[Lem. 1]{stehle_fieker_LLL} that we choose an LLL-reduced integral basis 
$\omega_1,\cdots, \omega_d$ of $\OK$ satisfying 
$$\max_i\|\omega_i\|\leq \sqrt{d}2^{d^2/2}\sqrt{|\Delta_K|}.$$
Then, we have 
\begin{align*}
\forall i\leq d, |x|_i &= |\sigma_i(x)|= \left| \sum_{j\leq d} |x_j| \sigma_i(\omega_j)\right|\\
&\leq d\left(\max_i |x_i|\right) \|\omega_j\|\\ 
&\leq \left(\max_i |x_i|\right) d^{3/2}2^{d^2/2}\sqrt{|\Delta_K|}.
\end{align*}
Therefore, $\log\left(\|x\|\right)\leq S(x)+ d\log\left(d^{3/2}2^{d^2/2}\sqrt{|\Delta_K|}\right)$. On the other hand, we know from~\cite[Lem. 2]{stehle_fieker_LLL} that for our choice of an integral basis 
of $\OK$, we have
$$\forall x\in\OK, \ S(x)\leq d\log\left(2^{3d/2}\|x\|\right).$$
\end{proof}

\subsection{Cost model}

\begin{proposition}
Let $\alpha,\beta\in K$ such that $S(\alpha),S(\beta)\leq B$, then the following holds:
\begin{enumerate}
 \item $\alpha + \beta$ can be computed in $\tilde{O}(dB)$
 \item $\alpha\beta$ can be computed in $\tilde{O}\left(d^2(B + d^3 + d\log|\Delta_K|)\right)$
 \item $\frac{1}{\alpha}$ can be computed in $\tilde{O}\left(d^{\omega-1}(B + d^3 + d\log|\Delta_K|)\right)$,
\end{enumerate}
where $\tilde{O}$ denotes the complexity whithout the logarithmic factors.
\end{proposition}

\begin{proof}
Let $x,y\in\OK$ and $a,b\in\Z_{>0}$ such that $\alpha = x/a$ and $\beta = y/b$. The first step of computing 
$\alpha + \beta$ consists of reducing them to the same denominator. This takes a time bounded by $\tilde{O}(dB)$.
Then the addition of the numerators takes 
$\tilde{O}(dB)$, as well as 
and the simplification by the GCD of the denominator and the $d$ coefficients.

For $i,j,k\leq d$, let $a^{(k)}_{i,j}$ be such that $\omega_i\omega_j = \sum_{k\leq d}a^{(k)}_{i,j}\omega_k$. 
From~\cite[Lem. 1]{stehle_fieker_LLL}, we know that $\forall i,\|\omega_i\|\leq \sqrt{d}2^{d^2/2}\sqrt{|\Delta_K|}$, 
and thus
$$\forall i,j\|\omega_i\omega_j\| \leq \| \omega_i \| \|\omega_j\| \leq d2^{d^2}|\Delta_K|.$$
Therefore, from Proposition~\ref{prop:coeff_norm}, we have $\forall i,j,k, \log\left(|a^{(k)}_{i,j}|\right)\leq \tilde{O}(d^2 + \log|\Delta_K|).$
Then, if $x = \sum_{i\leq d}b_i\omega_i$ and $y = \sum_{j\leq d}c_j\omega_j$, we first need to compute $b_ic_j$ for every 
$i,j\leq d$, which takes time $d^2\Mlt(B/d)$. Then, we compute $(b_ic_j)a^{(k)}_{i,j}$ for $i,j,k\leq d$, which takes 
$\tilde{O}(d^3\Mlt(2B/d + d^2 +\log|\Delta_K|)$. Then for each $k\leq d$, we compute 
$\sum_{i,j}b_ic_ja^{(k)}_{i,j}$, which is in $\tilde{O}(d(B/d + d^2 +\log|\Delta_K|))$. Finally, the multiplication of 
the denominators is in $\Mlt(B)$, and the simplification of the numerator and denominator takes 
$\tilde{O}(d\Mlt(B/d + d^2 + \log|\Delta_K|)).$

To invert $x = \sum_i b_i \omega_i$, we first define $A := (d_{j,k})_{j,k\leq d}$ by $d_{j,k}:= \sum_ib_ia^{(k)}_{i,j}$, 
and notice that 
$$\forall i,\ x\omega_i = \sum_ib_i\left( \sum_{k\leq d}a^{(k)}_{i,j}\omega_k\right) = \sum_{k\leq d} d_{j,k}\omega_k.$$
Inverting $x$ boils down to finding $x_1,\cdots,x_d\in\Q$ such that $\sum_i xx_i\omega_i = 1$. It can be achieved by solving 
$$XA = (1,0,\cdots,0).$$
We derive the complexity of this step by noticing that $\log|A|\leq 2^{B/d + d^2 +3d/2}d|\Delta_K|$. From Hadamard's inequality, 
we know that the numerator and the denominator of $x_i$ are bounded by 
$$d^{d/2}|A|^{d}\leq 2^{d^3 +3d^2/2 + B}d^{3d/2}|\Delta_K|^d.$$
Multiplying all numerators by $a$ where $\alpha = x/a$ costs 
$$\tilde{O}\left( d\Mlt(d^3  + B +d\log(|\Delta_K|)\right),$$ 
while reducing the $ax_i$ to the same denominator and simplifying the expression can be done in 
$$\tilde{O}\left(d(d^3 + B + d\log(|\Delta_K|))\right).$$
As $\omega \geq 2$, the complexity of the inversion is in fact dominated by the resolution of the linear system.
\end{proof}

\begin{proposition}
Let $\ag$ and $\bg$ be fractional ideals of $K$ such that $S(\ag),S(\bg)\leq B$, then the following holds:
\begin{enumerate}
 \item $\ag + \bg$ can be computed in $\tilde{O}(d^{\omega+1}B)$,
 \item $\ag\bg$ can be computed in $\tilde{O}(d^3(d^4 + d^2\log|\Delta_K| + B))$,
 \item $1/\ag$ can be computed in $\tilde{O}\left( d^{2\omega}(d^4 + d^2\log|\Delta_K| + B)\right)$.
\end{enumerate}
\end{proposition}

\begin{proof}
Let $A,C\in\Z^{d\times d}$ in HNF form and $a,c\in\Z_{>0}$ such that $\ag = \frac{1}{a}\left(\sum_{i\leq d}\Z A_i\right)$ 
and $\bg = \frac{1}{c}\left(\sum_{i\leq d}\Z C_i\right)$, where $A_i$ denotes the $i$-th row of $A$. Adding $\ag$ and $\bg$ 
is done by computing the HNF of $\left(\frac{cA}{aC}\right)$ and reducing the denominator. The complexity is bounded by the 
one of the HNF which is in $\tilde{O}(d^{\omega+1}B)$ since $\log|cA|,\log|aC|\leq B + B/d^2$.

Let $\gamma_1,\cdots,\gamma_d$ and $\delta_1,\cdots,\delta_d$ be integral elements such that 
\begin{align*}
\ag &= \frac{1}{a}\left( \Z\gamma_1 + \cdots + \Z\gamma_d \right) \\
\bg &= \frac{1}{b}\left( \Z\delta_1 + \cdots + \Z\delta_d \right) 
\end{align*}
for $a,b\in\Z_{>0}$. We first compute $\gamma_i\delta_j$, which takes 
$$\tilde{O}\left( d^3( S(\ag) + d^4 + d^2\log|\Delta_K|)\right).$$
Their size satisfies $S(\gamma_i\gamma_j)\leq \tilde{O}\left( d^3 + d\log|\Delta_K| + \frac{S(\ag)}{d}\right)$. Then, 
we compute the HNF basis of the $\Z$-module generated by the $\gamma_i\delta_j$, which costs
$$\tilde{O}\left( d^\omega ( d^4 + d^2\log|\Delta_K| + S(\ag))\right),$$
and we finally perform $d^2$ gcd reduction involving the product of the denominators which is bounded 
by $\tilde{O}(B)$.

Finally, we know from~\cite[4.8.4]{cohen} that finding the inverse of $\ag$ consists of calculating a basis of 
the nullspace of a matrix $D\in\Z^{(d^2+d)\times d^2}$ 
satisfying $\log|D|\leq \tilde{O}(d^2 + \log|\Delta_K| +B/d^2)$, and returning the HNF of its left $d\times d$ minor $U$. 
By using~\cite[Prop. 6.6]{Sto_phd}, we find such a nullspace $M\in\Z^{d\times d^2}$ satisfying $|M|\leq d(\sqrt{d}|D|)^{2d}$ 
in expected time bounded by 
$$\tilde{O}\left( d^{2+2\omega}\log|D|\right)\leq\tilde{O}\left( d^{2\omega}(d^4 + d^2\log|\Delta_K| + B)\right).$$
The HNF of $U$ has complexity bounded by $\tilde{O}(d^{\omega + 1}\log|M|)\leq \tilde{O}(d^{2 + \omega}\log|D|)$.
\end{proof}

\begin{proposition}
Let $\alpha\in K$, a fractional ideal $\ag\subseteq K$ and $B_1,B_2$ such that 
$S(\ag)\leq B_1$ and $S(\alpha)\leq B_2$, then $\alpha\ag$ can be computed in 
expected time bounded by 
$$\tilde{O}\left(d^{\omega}\left(d^3 + d\log|\Delta_K| +  \frac{B_1}{d} + B_2 \right)\right).$$
\end{proposition}

\begin{proof}
Let $x\in\OK$ and $a\in\Z_{>0}$ such that $\alpha = x/a$ and let  $k\in\Z_{>0}$ and $\gamma_1,\cdots,\gamma_d$ 
be  an HNF basis for $\ag$. Then, $(x\gamma_i)_{i\leq d}$ is a $\Z$-basis for $(x)\ag$. 
We perform $d$ multiplications $x\gamma_i$ where $S(\gamma_i)\leq B_1/d$ and $S(x)\leq B_2$. This costs
$$\tilde{O}\left( d^3\left( \frac{B_1}{d} + B_2 + d^3 + d\log|\Delta_K|\right) \right).$$
Then, from Corollary~\ref{cor:S_xy}, we know that 
$$S(x\gamma_i)\leq \tilde{O}\left(d^3 + d\log|\Delta_K| + S(x) + S(\gamma_i)\right).$$
Therefore, computing the HNF of the resulting matrix of entries bounded by $S(x\gamma_i)/d$ takes 
$$\tilde{O}\left(S(x\gamma_i) d^{\omega}\right)\leq \tilde{O}\left(d^{\omega}\left(d^3 + d\log|\Delta_K| + S(x) + S(\gamma_i)\right)\right).$$
Finally, we multiply the denominators and reduce them by successive GCD computations in time 
$\tilde{O}(dS(x\gamma_i)).$ 
\end{proof}

\subsection{Reduction modulo a fractional ideal}

\setcounter{proposition}{6}
\begin{proposition}
Let $x\in K$ and $\ag$ be a fractional ideal of $K$, then Algorithm~\ref{alg:reduction} returns $\overline{x}$ such that 
$x - \overline{x}\in\ag$ and 
$$\|\overline{x}\| \leq d^{3/2}2^{d/2}\Nm(\ag)^{1/d}\sqrt{|\Delta_K|}.$$
\end{proposition}

\begin{proof}
The LLL~\cite{LLL} algorithm allows us to compute a basis $(r_j)_{j\leq d}$ for $I$ that satisfies
$$\|r_j\|\leq 2^{d/2}\sqrt{d}\Nm(I)^{1/d}\sqrt{|\Delta_K|}.$$ 
The same holds for a fractional ideal $\ag$ of $K$ by multiplying the above relation by the denominator of $\ag$. 
Then, as $\lfloor x_j \rceil r_j\leq 1$, we see that 
$$\|\overline{x}\|\leq d\max_j\|r_j\|\leq d^{3/2}2^{d/2}\Nm(\ag)^{1/d}\sqrt{|\Delta_K|}.$$
\end{proof}

\subsection{The HNF}

At the end of Algorithm~\ref{alg:HNF}, we obtain a pseudo-basis 
$[(B_i)_{i\leq n},(\bg_i)_{i\leq n}]$ such that 
$$\forall i\leq n\  \bg_iB_i \subseteq M + \g e_i,$$
where $e_i := (0,0,\cdots,1,0,\cdots,0)$ is the $i$-th vector of the canonical basis of $K^n$. However, 
the determinant of $i\times i$ minors is preserved modulo $\g$. 
Let $M_i\subseteq \OK^{n-i}$ be the $\OK$-module defined by 
$$\ag_1(a_{1,n-i},\cdots , a_{1,n}) + \cdots + \ag_n(a_{n,n-i},\cdots,a_{n,n}),$$
and $\g(M_i)$ its determinantal ideal. The operations performed at Step 6 to 10 in Algorithm~\ref{alg:HNF} preserve 
$\g(M_i)$ while after Step~11, our pseudo-basis $[(B_i)_{i\leq n},(\bg_i)_{i\leq n}]$ only defines a module $M'\subseteq\OK^n$
satisfying 
$$\g(M'_i) + \g = \g(M_i) + \g.$$
This property is the equivalent of the integer case when the HNF is taken modulo a multiple $D$ of the determinant of the 
lattice. To recover the ideals $\cg_i$ of a pseudo-HNF of $M$, we first notice that 
\begin{align*}
\forall i,\g(M_i') + \g = \g(M_i)+g
&= \cg_{n-i}\cdots\cg_n + \g \\
&= \cg_{n-i}\cdots\cg_n + \cg_1\cdots \cg_n \\
&= \cg_{n-i}\cdots\cg_n.
\end{align*}
On the other hand, $\g(M'_i) + \g = \bg_{n-i}\cdots\bg_n + \g$. Thus, we have 
$$\forall i,\ \bg_{n-i}\cdots\bg_n +\g= \cg_{n-i}\cdots\cg_n,$$
which allows us to recursively recover the $\cg_i$ from the $(\bg_j)_{j\geq i}$ and $\g$. Indeed, as in the integer case, 
it boils down to taking 
$$\cg_i = \frac{\g}{\prod_{j > i}\cg_j} + \bg_i.$$
To do so, we keep track of $\g_i := \frac{\g}{\prod_{j > i}\cg_j}$ throughout Algorithm~\ref{alg:Euclidian}
that reconstructs the actual pseudo-HNF from its modular version given by Algorithm~\ref{alg:HNF}. At each step we set 
$$\cg_i\leftarrow \bg_i + \g_{i}.$$
This replacement of the ideals in the pseudo-basis defining our module impacts the corresponding vectors 
in $K^n$ as well. In particular, we require that the diagonal elements all be 1. Do ensure thus, we find $u\in \bg_i\cg^{-1}_i,\ v\in \g_{i}\cg^{-1}_i$ such that $u + v = 1$ which implies that 
$$\cg_i(uB_i + ve_i)\subseteq \bg_iB_i + \g_ie_i,$$
where the $i$-th coefficient of $uB_i + ve_i\in K^n$ is 1 and the coefficient of index $j>i$ in $uB_i + ve_i$
are 0. Then we set 
$$W_i\leftarrow uB_i\bmod \g_i\cg^{-1}_i,$$
and observe that $\sum_i \cg_iW_i \subseteq M$. These $\OK$-modules have the same determinantal ideal, 
and as in the integer case, we can prove that it is sufficient to ensure that they are equal.

$uB_i + ve_i = W_i + d_i$ where the coefficients of $d_i\in\left(\g_i/\cg_i\right)^n$ of 
index $j > i$ are 0. The vector $d_i$ satisfies $\cg_id_i\subseteq \g_id'_i$ where $d'_i\in\OK^n$ with coefficients $j>i$ equal to 0. This allows us to state that 
$$\cg_i W_i \subseteq \bg_iB_i + \g_ie_i  + \cg_id_i \subseteq M + \g_ie_i  + \g_id'_i \subseteq M + \g_iD_i,$$
where the coefficients of $D_i\in\OK^n$ of index $j>i$ equal 0. We now want to prove that 
$\cg_iW_i\subseteq M$. To do this, we prove that $\g_iD_i\subseteq M$. 

\begin{lemma}\label{lem:1}
Let $M = \ag_1A_1 + \cdots \ag_nA_n\in\OK^n$, then we have 
$$\g(M)\OK^n \subseteq M$$ 
\end{lemma}

\begin{proof}
We can prove by induction that if $[(B_i),(\bg_i)]$ is a pseudo-HNF basis of $M$, then 
$$\forall i,\ \g_1\cdots\g_i e_i \subseteq M,$$
where $e_i$ is the $i$-th vector of the canonical basis of $\OK^n$. Our statement immediatly follows.
\end{proof}

We now consider the intersection $N_i$ of our module $M\subseteq \OK^n$ with $\OK^i$. Note that with the 
previous definitions, we have in particular $M = N_i \oplus M_i$. 

\begin{lemma}
Let $i\leq n$ and $D\in\OK^n$ a vector whose entries of index $j>i$ are 0. Then we have 
$$\g_i D \subseteq M.$$
\end{lemma}

\begin{proof}
From Lemma~\ref{lem:1}, we know that $\g_i\OK^i \subseteq N_i$. If $D_i\in\OK^i$ is the first $i$ 
coordinates of $D$, then $\g_i D_i\subseteq N_i$, and as the last $n-i$ coordinates of $D$ are 0, we have 
$$\g_i D \subseteq M.$$
\end{proof}
The module generated by the pseudo-basis $[(W_i),(\cg_i)]$ computed by Algorithm~\ref{alg:Euclidian} is 
a subset of $M$. We ensured that its determinantal ideal $\prod_i\cg_i$ equals the determinantal ideal 
$\g$ of $M$. Let us prove that it is sufficient to ensure that 
$$\cg_1 W_1 + \cdots + \cg_n W_n = M.$$

\begin{lemma}\label{lem:3}
Let $M = \sum_{i\leq n} \ag_i A_i$ and $M' = \sum_{i\leq n} \bg_i B_i$ two $n$-dimensional $\OK$-modules 
such that $M'\subseteq M$ and $\g(M') = \g(M)$. Then necessarily 
$$M = M'.$$
\end{lemma}

\begin{proof}
Let $[(W_i),(\cg_i)]$ be a pseudo-HNF for $M$, and $[(W'_i),(\cg'_i)]$ a pseudo-HNF for $M'$. By assumption, we have $\prod_i\cg_i= \prod_i\cg'_i$, and $M'\subseteq M$. As both matrices $W$ and $W'$ have a lower
triangular shape, it is clear that 
\begin{equation}
 \forall i, \ \sum_{j\leq i}\cg'_jW'_j \subseteq \sum_{j\leq i} \cg_jW_j.\label{eq:mod_inc}
\end{equation}
As the diagonal coefficients of both $W$ and $W'$ are 1, we see by looking at the inclusion in the coefficient $i$ of~\eqref{eq:mod_inc} that $\cg'_i\subseteq \cg_i$. Then as $\g(M) = \g(M')$, we have 
$$\forall i \cg_i = \cg'_i.$$
Now let us prove by induction that 
\begin{equation}
 \forall i,\ \cg_iW_i \subseteq \cg_1W'_1 + \cdots + \cg_iW'_i.\label{eq:rec_mod_inc}
\end{equation}
This assertion is clear for $i=1$ since $W_1 = W'_1 = e_1$. Then, assuming~\eqref{eq:rec_mod_inc} for 
$1,\cdots,i-1$, we first use the fact that 
$$\cg_iW'_i \subseteq \cg_1W_1 + \cdots + \cg_i W_i.$$
In other words, $\forall c'_i\in\cg_i$, $\exists (c_1,\cdots,c_i)\in\cg_1\times\cdots\times\cg_i$ such that 
$$c'_i(w'_{i,1},\cdots,w'_{i,i-1},1) = \left(\sum_{1\leq j\leq i} c_jw_{j,1} , \cdots , c_{i}w_{i,i-1} + c_{i-1} , c_i\right).$$ 
In particular, $c_i = c'_i$, which allows us to state that $\forall c_i\in\cg_i$, 
$\exists (c_1,\cdots,c_{i-1})\in\cg_1\times\cdots\times\cg_{i-1}$ such that 
\begin{align*}
c_i w_{i,i-1} &= c_{i-1}  + c_i w'_{i,i-1} \\
c_i w_{i,i-2} &= c_{i-2}  + c_{i-1} w_{i-1,i-2} + c_iw'_{i,i-2} \\
\vdots\ \  &=\ \  \vdots \\
c_i w_{i,1} &= c_{1} + \cdots + c_{i-1} w_{i-1,1} + c_iw'_{i,1}.
\end{align*}
This shows that 
$$\cg_iW_i \subseteq \cg_1W_1 + \cdots + \cg_{i-1}W_{i-1} + \cg_iW'_i,$$
and since we have $\forall j<i,\ \cg_jW_i\subseteq \sum_{j<i}\cg_jW'_j$, we obtain the desired result.
\end{proof}

Lemma~\ref{lem:3} is a generalization of the standard result on $\Z$-modules stating that if $L'\subseteq L$
and $\det(L) = \det(L')$, then $L = L'$. Although implied in~\cite[Chap. 1]{cohen2}, Lemma~\ref{lem:3} is 
not stated, nor proved in the litterature. Yet, it is essential to ensure the validity of 
Algorithm~\ref{alg:Euclidian}. 

\setcounter{proposition}{8}
\begin{proposition}
The $\OK$-module defined by the pseudo-basis $[(W_i),(\cg_i)]$ obtained by applying lgorithm~\ref{alg:Euclidian} to the pseudo-HNF of $M$ modulo $\g(M)$ satisfies 
$$\cg_1 W_1 + \cdots + \cg_n W_n = M.$$
\end{proposition} 

\balancecolumns
\end{document}